\documentclass[11pt]{article}
\usepackage{amsmath,amssymb,amsthm,xspace,algorithm,fullpage}
\usepackage[noend]{algorithmic}

\usepackage[colorlinks=true]{hyperref}

\newtheorem{theorem}{Theorem}[section]
\newtheorem{lemma}[theorem]{Lemma}

\newtheorem{proposition}[theorem]{Proposition}

\newtheorem{corollary}[theorem]{Corollary}

\newcommand{\bbC}{\mathbb{C}}
\newcommand{\bbD}{\mathbb{D}}
\newcommand{\bbF}{\mathbb{F}}

\newcommand{\bbR}{\mathbb{R}}
\newcommand{\bbZ}{\mathbb{Z}}
\newcommand{\caA}{\mathcal{A}}
\newcommand{\caP}{\mathcal{P}}
\newcommand{\E}{\mathop{\mathbf{E}}}
\newcommand{\tr}{\mathrm{tr}}
\newcommand{\iunit}{\sqrt{-1}}
\newcommand{\dist}{\mathrm{dist}}
\newcommand{\id}{\mathrm{id}}

\newcommand{\GL}{\mathrm{GL}}
\newcommand{\Skew}{\mathrm{Skew}}
\newcommand{\diag}{\mathrm{diag}}

\title{Testing Properties of Functions on Finite Groups}

\author{
  Kenta Oono\\
  Preferred Networks, Inc.\\
  \texttt{oono@preferred.jp}
  \and
  Yuichi Yoshida\thanks{
    Supported by JSPS Grant-in-Aid for Young Scientists (B) (No.~26730009), MEXT Grant-in-Aid for Scientific Research on Innovative Areas (No.~24106001), and JST, ERATO, Kawarabayashi Large Graph Project.
  }\\
  National Institute of Informatics, and \\
  Preferred Infrastructure, Inc. \\
  \texttt{yyoshida@nii.ac.jp}
}

\begin{document}

\maketitle

\begin{abstract}
  We study testing properties of functions on finite groups.
  First we consider functions of the form $f:G \to \bbC$, where $G$ is a finite group.
  We show that conjugate invariance, homomorphism, and the property of being proportional to an irreducible character is testable with a constant number of queries to $f$,
  where a character is a crucial notion in representation theory.
  Our proof relies on representation theory and harmonic analysis on finite groups. 
  Next we consider functions of the form $f: G \to M_d(\bbC)$,
  where $d$ is a fixed constant and  $M_d(\bbC)$ is the family of $d$ by $d$ matrices with each element in $\bbC$.
  For a function $g:G \to M_d(\bbC)$,
  we show that the unitary isomorphism to $g$ is testable with a constant number of queries to $f$,
  where we say that $f$ and $g$ are unitary isomorphic if there exists a unitary matrix $U$ such that $f(x) = Ug(x)U^{-1}$ for any $x \in G$.
\end{abstract}


\section{Introduction}\label{sec:intro}

In \emph{property testing}~\cite{Rubinfeld:1996um,Blum:1993cn},
we want to decide whether the input function $f$ satisfies a predetermined property $\caP$ or ``far'' from it.
More specifically,
an algorithm is called a \emph{tester} for a property $\caP$ if,
given a query access to the input function $f$ and a parameter $\epsilon > 0$,
it accepts with probability at least $2/3$ when $f$ satisfies $\caP$,
and rejects with probability at least $2/3$ when $f$ is \emph{$\epsilon$-far} from $\caP$.
If a tester accepts with probability one when $f$ satisfies $\caP$, then it is called a \emph{one-sided error tester}.
The definition of $\epsilon$-farness depends on the model,
but for the case of Boolean functions,
we say that a function $f:\bbF_2^n \to \{-1,1\}$ is $\epsilon$-far from $\caP$ if the \emph{distance} $\dist(f,g) := \Pr_x[f(x) \neq g(x)]$ of $f$ and $g$ is at least $\epsilon$ for any function $g:\bbF_2^n \to \{-1,1\}$ satisfying $\caP$.
Here $-1$ and $1$ are corresponding to \textbf{true} and \textbf{false}, respectively.
The efficiency of the tester is measured by the number of queries to $f$, called the \emph{query complexity}.
We say that a property $\caP$ is \emph{constant-query testable} if there is a tester with query complexity depending only on $\epsilon$ (and not on $n$ at all).

The study of testing properties of Boolean functions, or more generally, functions on finite fields was initiated by Rubinfeld and Sudan~\cite{Rubinfeld:1996um},
and then subsequently many properties have been shown to be constant-query testable~\cite{Blum:1993cn,Bhattacharyya:2010kw,Green:2005iv,Yoshida:2012dm}.
To incorporate the algebraic structure of the finite field,
Kaufman and Sudan~\cite{Kaufman:2008jm} asked to study \emph{affine-invariant properties},
that is, properties $\caP$ such that, if $f:\bbF_p^n \to \{0,1\}$ satisfies $\caP$, then $f \circ A$ also satisfies $\caP$ for any bijective affine transformation $A$.
A lot of progress has been made on the study of the constant-query testability of affine-invariant properties~\cite{Bhattacharyya:2012ud,Hatami:2013ux},
and finally (almost) complete characterizations of constant-query testability were achieved~\cite{Bhattacharyya:2013ii,Yoshida:2014tq}.
For further details on function property testing, refer to~\cite{Bhattacharyya:2013gu,Ron:2009hh} for surveys.

Besides finite fields,
functions over finite groups such as the cyclic group and the permutation group are also objects that naturally appear in various contexts, e.g., circuit complexity~\cite{Allender:1998hf}, computational learning~\cite{Wimmer:2010jw}, and machine learning~\cite{Huang:2009vk}.
Despite its importance, there are only a few works on testing properties on functions over finite groups~\cite{BenOr:2007cg,Blum:1993cn}, and extending this line of research is the main focus of the present paper.
More specifically,
we consider testing properties of functions $f$ of the form $f:G \to \bbD$, and more generally, $f:G \to \bbD(d)$,
where $G$ is a finite group, $\bbD = \{z \in \bbC \mid |z| \leq 1\}$ is the unit disk, and $\bbD(d) = \{A \in M_d(\bbC) \mid \|A\|_F \leq 1 \}$ is the set of $d$ by $d$ matrices with Frobenius norm at most one.
Note that $\bbD(1) = \bbD$.
We regard $d$ as a constant.
The reason that we use $\bbD$ and $\bbD(d)$ is that they are maximal sets closed under multiplication.
Below, we get into the details of these two settings.

\paragraph{Testing properties of functions of the form $f:G \to \bbD$:}

We define the \emph{distance} between two functions $f,g:G \to \bbD$ as $\dist(f,g) = \frac{1}{2}\|f - g\|_2$,
where $\|f\|_2 := \sqrt{\E_{x \in G}|f(x)|^2}$ is the $L_2$ norm.
Note that $\dist(f,g)$ is always in $[0,1]$.
We say that a function $f:G \to \bbD$ is $\epsilon$-far from a property $\caP$ if $\dist(f,g) \geq \epsilon$ for any function $g:G \to \bbD$ satisfying $\caP$.
We note that,
for $\{-1,1\}$-valued functions $f,g:G\to \{-1,1\}$,
we have $\Pr[f(x) \neq g(x)] = \epsilon$ if and only if $\dist(f,g) = \sqrt{\epsilon}$ holds.
Hence, we have a quadratic gap between our definition and the standard definition using the Hamming distance.
However, we adopt the $L_2$ norm as it is more friendly with our analysis.

We first show the following:
\begin{itemize}
\itemsep=0pt
\item Conjugate invariance, that is, $f(yxy^{-1}) = f(x)$ for any $x,y \in G$, is one-sided error testable with $O(1/\epsilon^2)$ queries.
\item Homomorphism, that is, $f(x)f(y) = f(xy)$ for any $x,y \in G$, is one-sided error testable with $O(1/\epsilon^2 \log(1/\epsilon))$ queries.
\end{itemize}
We show the constant-query testability of conjugate invariance by a simple combinatorial argument.

When $G = \bbF_2^n$, then homomorphism is often called linearity and intensively studied in the area of property testing~\cite{Bellare:1996iv,BenSasson:2003dn,Blum:1993cn,Shpilka:2006fp}.
Indeed in this case, a function $f:\bbF_2^n \to \bbD$ is homomorphism if and only if $f(x) = \chi_S(x) := (-1)^{\sum_{i \in S}x_i}$ for some $S \subseteq \{1,\ldots,n\}$.

The case $G  = S_n$, the permutation group of order $n$, is easy to understand.
Let $f:S_n \to \bbD$ be a function on the permutation group.
If $f$ is conjugate invariant,
then the value of $f(\pi)$ only depends on the cycle pattern of $\pi$.
If $f$ is homomorphism,
then $f$ is the all-zero function, the all-one function, or the function that returns the sign of the input permutation.


We note that Ben-Or~\emph{et~al.}~\cite{BenOr:2007cg} studied the constant-query testability of homomorphism from a finite group to another finite group.
The query complexity of their algorithm is $O(1/\epsilon)$, where the distance is measured by the Hamming distance.
Their algorithm and analysis by a combinatorial argument extends to our setting, in which the range is $\bbD$ and the distance is measured by the $L_2$ norm, with query complexity $O(1/\epsilon^2)$.
In this sense, our result on homomorphism is not new.
Nevertheless, we prove it again using harmonic analysis over finite groups.
By doing so, we can generalize it for testing other properties. 






To describe our next result,
we need to introduce the basic of representation theory.
A \emph{representation} of a group $G$ is a homomorphism $\varphi$ of the form $\varphi:G \to M_{d}(\bbC)$ for some integer $d$.
In particular, we study the family of \emph{irreducible representations}, where any representation can be described as the direct sum of irreducible representations.
We mention that irreducible representations $\varphi:G \to M_{d}(\bbC)$
can be chosen as \emph{unitary}, that is, $\varphi(x)$ is unitary for all $x \in G$.
The \emph{character} of a representation $\varphi$ is the function $\chi_\varphi(x) = \tr(\varphi(x))$,
where $\tr(\cdot)$ denotes the trace of a matrix.
The character carries the essential information about the representation in a more condensed form and is intensively studied in \emph{character theory}.
The character of an irreducible representation is called an \emph{irreducible character}.
For example,
every irreducible character of $\bbF_2^n$ is of the form $\chi_S(x) = (-1)^{\sum_{i \in S}x_i}$ for some $S \subseteq \{1,\ldots,n\}$, which is often called a \emph{character} in Fourier analysis of Boolean functions.
If $G$ is \emph{abelian}, that is, commutative, then representations of $G$ always map to one-dimensional matrices,
and hence a representation and its corresponding character coincide.
However, this is not the case when $G$ is not abelian, which makes the analysis more involved.
We show the following:
\begin{itemize}
\item The property of being proportional to an irreducible character, that is, $f = c \chi_\varphi$ for some $c \in \bbC$ and irreducible representation $\varphi$, is testable with $O(1/\epsilon^8 \log^2(1/\epsilon))$ queries.
\end{itemize}
The reason that we do not consider irreducible characters themselves is that irreducible characters may take values outside of $\bbD$.
In particular, $\chi_\varphi(1) = d$ holds for the identity element $1 \in G$ and the dimension $d$ of $\varphi$.

When $G = \bbF_2^n$, then irreducible characters coincide with linear functions, that is, $\chi_S$ for some $S \subseteq \{1,\ldots,n\}$.
Hence, testing the property of being proportional to an irreducible character can be seen as another generalization of linearity testing.

The form of irreducible characters is quite complicated in general.
When $G = S_n$,
however,
its combinatorial interpretation via Young Tableau is well studied~\cite{Biane:2003hw,Kerov:2000aa,Stanley:2003tv} (though still complicated to state here).

\paragraph{Testing properties of functions of the form $f:G \to \bbD(d)$:}
Since the representations of a group $G$ are matrix-valued,
it is natural to consider testing properties of functions of the form $f:G \to \bbD(d)$.
For two functions $f,g:G \to \bbD(d)$,
we define $\dist(f,g) = \frac{1}{2}\sqrt{\E_x\|f(x) - g(x)\|_F^2}$. 
Note that this is indeed a metric and matches the previous definition of distance when $d = 1$.

Let $U(d)$ be the set of $d$ by $d$ unitary matrices with each element in $\bbC$.
We show the following:
\begin{itemize}
\item Unitary equivalence to $g:G \to \bbD(d)$, that is, $f = UgU^{-1}$ for some unitary matrix $U \in U(d)$, is testable with $(d^{3/2}/\epsilon)^{O(d^2)}$ queries.
\end{itemize}
Here $g$ is a parameter of the problem and not a part of the input.
Unitary equivalence is an important notion when studying representations since the irreducibility of a representation $\varphi$ is preserved by the transformation $\varphi \mapsto U\varphi U^{-1}$ for a unitary matrix $U$.

Our tester samples unitary matrices from the Haar measure, a fundamental tool in the representation theory of Lie Groups, and then checks whether $g$ becomes close to $f$ by applying these unitary matrices.

Arguably the simplest property of matrix-valued functions is again homomorphism.
However, homomorphism is known to be constant-query testable by a combinatorial argument~\cite{BenOr:2007cg},
and the harmonic analysis does not facilitate the analysis.
Therefore we do not study homomorphism of matrix-valued functions in this paper.



\paragraph{Related work:}
There are a number of works on testing whether a function on a finite group is a homomorphism.
Blum~\emph{et~al.}~\cite{Blum:1993cn} gave a tester (the \emph{BLR tester}) for homomorphism of functions on a finite group.
However, the number of queries depends on the number of generators of the group,
which may depend on the size of the group in general.
Ben-Or~\emph{et~al.}~\cite{BenOr:2007cg} gave another algorithm without this dependency.
Bellare~\emph{et~al.}~\cite{Bellare:1996iv} gave a Fourier-analytic proof of the BLR tester when the domain and the range are $\bbF_2^n$ and $\bbF_2$, respectively.
Our tester for homomorphism can be seen as a generalization of their analysis to general groups.
There has been an interest in improving various parameters of homomorphism testing results, due to their applications in the construction of probabilistically checkable proof (PCP)~\cite{Arora:1998dga}.
Bellare~\emph{et~al.}~\cite{Bellare:1996iv} gave an almost tight connection between the distance to homomorphism and the rejection probability of the BLR tester.
Ben-Sasson~\emph{et~al.}~\cite{BenSasson:2003dn} and Shpilka and Wigderson~\cite{Shpilka:2006fp} reduced the number of random bits required by the  test as it affects the efficiency of the proof system and in turn the hardness of approximation results that one can achieve using the proof system.
Rubinfeld~\cite{Rubinfeld:2006jy} studied properties of a function on a finite group that are defined by functional equations and gave a sufficient condition of constant-query testability.

Using the $L_p$ norm for $p \geq 1$ as a distance measure in property testing is recently systematically studied by Berman~\emph{et~al.}~\cite{Berman:2014kg}.
One of their motivations is exploring the connection of property testing with learning theory and approximation theory.
For this purpose, the $L_p$ norm is more favorable than the Hamming distance because, in learning theory and approximation theory, we typically measure errors in the $L_p$ norm for $p = 1$ or $2$.
Indeed, several lower bounds and upper bounds for property testing in the $L_p$ norm were shown using this connection.
See~\cite{Berman:2014kg} for more details.

Representation theory is one of the most important areas in modern mathematics.
Representation theory itself is intensively studied and it is also used as an analytical tool in harmonic analysis, invariant theories, and modular theory.
Representation theory have been used in various problems of theoretical computer science such as constructing pseudorandom objects~\cite{Alon:2013ci,Kassabov:2007ki}, circuit complexity~\cite{Allender:1998hf}, communication complexity~\cite{Raz:1995ib}, computational learning~\cite{Wimmer:2010jw}, machine learning~\cite{Huang:2009vk}, and quantum property testing~\cite{ODonnell:2015va}.

\paragraph{Organization:}
In Section~\ref{sec:pre}, we introduce representation theory, harmonic analysis on finite groups, and the Haar measure in more detail.
We discuss the testability of conjugacy invariance in Section~\ref{sec:conjugate}.
We show that homomorphism and the property of a being proportional to an irreducible character are constant-query testable in Sections~\ref{sec:homo} and~\ref{sec:nrc}, respectively.
Section~\ref{sec:unitary} is devoted to testing unitary equivalence.

\section{Preliminaries}\label{sec:pre}

For an integer $n \geq 1$, $[n]$ denotes the set $\{1,2,\ldots,n\}$.
Let $\delta_{ij}$ be \emph{Kronecker's delta}, that is, $\delta_{ij} = 1$ if $i = j$ and $\delta_{ij} = 0$ if $i \neq j$.
For a matrix $M$, we denote its $(i, j)$-th element by $M_{ij}$.
We write the real and imginary part of a complex number $z$ as $\Re z$ and $\Im z$, respectively (hence $z = \Re z + \iunit \Im z$).
For a complex number $z$, $\overline{z}$ denotes its conjugate.

We frequently use the following lemma.
\begin{lemma}\label{lem:estimate}
  Let $f: G \to \bbD$ be a function for some finite group $G$.
  For any $\epsilon > 0$,
  with probability at least $1-\delta$,
  we can compute an estimate $z$ of $\E_{x \in G}[f(x)]$ such that $|z - \E_{x \in G}[f(x)]| \leq \epsilon$ .
  The number of queries to $f$ is $O(1/\epsilon^2 \log 1/\delta)$.
\end{lemma}
\begin{proof}
  Let $c = \E_{x \in G}[f(x)]$.
  To estimate $c$,
  we sample $x_1,\ldots,x_s$ uniformly at random from $G$,
  where $s = \Theta(1/\epsilon^2 \log 1/\delta)$.
  Then, we output $\widetilde{c} := \frac{1}{s}\sum_{i \in [s]}f(x_x)$.
  Clearly, the query complexity is $O(1/\epsilon^2 \log 1/\delta)$.

  We now show that $\widetilde{c}$ is indeed a good approximation to $c$.
  Since $|\Re(f(x))| \leq 1$,
  from Hoeffding's inequality,
  we have $|\Re(c) - \Re(\widetilde{c})| \leq \epsilon/2$ with probability at least $1-\delta/2$ by choosing the hidden constant in $s$ large enough.
  Similarly,
  we have $|\Im(c) - \Im(\widetilde{c})| \leq \epsilon/2$ with probability at least $1-\delta/2$.
  By the union bound, we have $|c-\widetilde{c}| \leq |\Re(c) - \Re(\widetilde{c})| + |\Im(c) - \Im(\widetilde{c})| \leq \epsilon$ with probability at least $1-\delta$.
\end{proof}



\subsection{Representation theory}
We introduce basic notions and facts in representation theory.
See, e.g.,~\cite{steinberg2011representation} for more details.

For a vector space $V$ over a field $\bbF$,
$\GL_{\bbF}(V)$ denotes the set of invertible linear transformations.
We only consider the case $\bbF = \bbC$ in this paper, and hence we omit the subscript for simplicity.

A \emph{representation} of $G$ is a pair $(\varphi,V)$ of a finite-dimensional vector space $V$ and a homomorphism $\varphi:G \to \GL(V)$,
that is, $\varphi(xy) = \varphi(x)\varphi(y)$ for every $x,y \in G$ and $\varphi(1)$ is the identity transformation for the identity element $1 \in G$.
For a representation $(\varphi,V)$, $V$ is called the \emph{representation space} of it.
When $V$ is clear from the context, we simply call $\varphi$ a representation.
The \emph{dimension} of a representation $(\varphi,V)$ is the dimension of $V$.
When $V$ is a finite-dimensional vector space, then we say that $(\varphi,V)$ is \emph{finite-dimensional representation}.
In our argument, we only need finite-dimensional representations.

We describe the decomposition of a representation into irreducible representations, which is a fundamental tool used in representation theory.
For a representation $(\varphi,V)$ and a subspace $W$ of $V$,
we say that $W$ is \emph{$G$-invariant} if $\varphi(G)W \subseteq W$.
If $W$ is a $G$-invariant space, then we can regard the range of $\varphi$ as $\GL(W)$, and hence we obtain a representation $(\varphi,W)$.
Note that $\{0\}$ and $V$ are $G$-invariant from the definition.
A representation $(\varphi,V)$ is called \emph{irreducible} if $\{0\}$ and $V$ are the only $G$-invariant spaces.
Note that a one-dimensional representation is always irreducible.
When $G$ is abelian, then we have the converse from Schur's Lemma,
that is, any irreducible representation is one-dimensional.
When $G$ is non-abelian, however, an irreducible representation might have dimension more than one.
This fact makes the analysis of algorithms for functions on a non-abelian group more involved.

Two representations $(\varphi,V)$ and $(\psi,W)$ of $G$ are \emph{equivalent} if there exists an invertible linear transformation $T:V \to W$ such that, for every $x \in G$, it holds that $\psi(x) \circ T = T \circ \varphi(x)$ .
We identify equivalent representations,
and we denote by $\widehat{G}$ the family of equivalence classes of irreducible representations.
It is known that there is a one-to-one correspondence between conjugacy classes of $G$ and $\widehat{G}$.

A representation $(\varphi,V)$ is \emph{unitary} if, for all $x \in G$, $\varphi(x)$ is a unitary transformation.
For any representation of $G$, there is an equivalent unitary representation.
Hence, we can take unitary representations as a complete system of representatives of $\widehat{G}$, and we identify it with $\widehat{G}$.
Since $G$ is finite, so is $\widehat{G}$.
For $\varphi \in \widehat{G}$, we denote the dimension of its representation space by $d_{\varphi}$.
In what follows,
we fix a basis of the vector space of each representation $(\varphi,V)$,
and we regard it as a homomorphism from $G$ to $M_{d_\varphi}(\bbC)$, where $d_{\varphi}$ is the dimension of $V$.



\subsection{Fourier analysis on non-abelian finite groups}

We regard the space of $\bbC$-valued functions of $G$ as an inner product space by defining $\langle f, g\rangle = \E_{x \in G}[f(x)\overline{g(x)}]$ for $f,g:G \to \bbC$.
The following fact is known.
\begin{lemma}[\cite{steinberg2011representation}]\label{lem:matrix-coefficients-forms-orthonormal-basis}
  For a finite group $G$, the set $\bigl\{\sqrt{d_{\varphi}} \varphi_{ij} \mid  \varphi \in \widehat{G}, i, j \in [d_{\varphi}]\bigr\}$ forms an orthonormal basis of the space of $\bbC$-valued functions of $G$.
\end{lemma}
Hence, we can decompose $f:G \to \bbC$ as
\[
  f(x)
  = \sum_{\varphi \in \widehat{G}}d_{\varphi}\sum_{i,j \in [d_{\varphi}]}\langle f, \varphi_{ij}\rangle \varphi_{ij}(x)
  = \sum_{\varphi \in \widehat{G}}d_{\varphi}\sum_{i,j \in [d_{\varphi}]}\widehat{f} (\varphi)_{ij}\varphi_{ij}(x),
\]
where $\widehat{f}(\varphi) \in M_{d_\varphi}(\bbC)$ is defined as $\widehat{f}(\varphi) = \E_{x \in G}[f(x) \overline{\varphi(x)}]$ and called the \emph{Fourier coefficient} of $\varphi$.
This decomposition is called the Fourier expansion of $f$.
Note that Fourier coefficients are matrix-valued functions.
The following is well known.
\begin{lemma}[\cite{steinberg2011representation}]
  Let $f,g:G \to \bbC$ be functions.
  Then, we have
  \begin{align*}
    \langle f, g \rangle &=
    \sum_{\varphi}d_{\varphi}\sum_{i,j \in [d_{\varphi}]}\widehat{f}(\varphi)_{ij}\overline{\widehat{g}(\varphi)_{ij}},
    \tag{Plancherel's identity}\\
    \|f\|_2^2 & =
    \sum_{\varphi}d_{\varphi}\sum_{i,j \in [d_{\varphi}]}|\widehat{f}(\varphi)_{ij}|^2.
    \tag{Parseval's identity}
  \end{align*}
\end{lemma}

\subsection{Class functions and characters}\label{subsec:characters}

For a representation $\varphi:G \to M_{d_\varphi}(\bbC)$,
the \emph{character} $\chi_\varphi:G \to \bbC$ of $\varphi$ is defined as $\chi_\varphi(x) = \tr(\varphi(x))$ for $x \in G$.
We say that a function $f:G \to \bbC$ is \emph{conjugate invariant} if $f(x) = f(yxy^{-1})$ for all $x, y \in G$.
A conjugate invariant function is sometimes called a \emph{class function}.
It is not hard to check that characters are conjugate invariant.
Indeed, the following fact is known.
\begin{lemma}
  For a finite group $G$, the set $\bigl\{\chi_{\varphi} \mid \varphi \in \widehat{G}\bigr\}$ forms an orthonormal basis of the space of $\bbC$-valued class functions of $G$.
\end{lemma}

Note that if a representation is one-dimensional, its character is identical to the original representation, hence is a homomorphism.
This is not the case in general.

The following lemma says that Fourier coefficients of a class function are always diagonal.
\begin{lemma}\label{lem:fourier-expansion-of-class-function}
  For any class function $f: G\to \bbC$, it holds that $\widehat{f}(\varphi) = \frac{\langle f, \chi_{\varphi} \rangle}{d_{\varphi}} I_{d_{\varphi}}$.
\end{lemma}


In order to prove Lemma~\ref{lem:fourier-expansion-of-class-function}, we need the following two auxiliary lemmas:
\begin{lemma}\label{lem:trace-of-fourier-coefficient}
  For a function $f:G \to \bbC$ and an irreducible representation $\varphi$, we have
  \begin{align*}
    \langle f, \chi_\varphi \rangle = \tr(\widehat{f}(\varphi)).
  \end{align*}
\end{lemma}
\begin{proof}
  \begin{align*}
    \langle f, \chi_\varphi \rangle
    & =
    \E_x \sum_\rho d_\rho \sum_{i,j}\widehat{f}(\rho)_{ij} \rho_{ij}(x) \sum_k \overline{\varphi_{kk}(x)}
    =
    \sum_\rho d_\rho \sum_{i,j}\widehat{f}(\rho)_{ij}\sum_k \langle \rho_{ij},\varphi_{kk}\rangle \\
    & =
    \sum_{k}\widehat{f}(\varphi)_{kk}
    =
    \tr(\widehat{f}(\varphi)).
    \qedhere
  \end{align*}
\end{proof}


\begin{lemma}\label{lem:distance-to-class-functions}
  Let $f:G \to \bbC$ be a function and $g:G\to \bbC$ be a class function.
  Then,
  \begin{align*}
    \langle f, g \rangle = \sum_{\varphi} \tr(\widehat{f}(\varphi))\overline{\tr(\widehat{g}(\varphi))}.
  \end{align*}
\end{lemma}
\begin{proof}
  Since $g$ is a class function, we can represent $g(x) = \sum_{\varphi}\langle g, \chi_\varphi\rangle \chi_\varphi(x)$.
  Now we have
  \begin{align*}
    \langle f, g \rangle
    & =
    \sum_{\varphi,\varphi'}d_{\varphi}\sum_{i,j} \widehat{f}(\varphi)_{ij} \overline{\langle g, \chi_\varphi'\rangle} \E_x [ \varphi_{ij}(x) \overline{\chi_{\varphi'}(x)}] \\
    & =
    \sum_{\varphi,\varphi'}d_{\varphi}\sum_{i,j} \widehat{f}(\varphi)_{ij} \sum_{i'}\overline{\widehat{g}(\varphi')_{i'i'}} \sum_{j'} \E_x [ \varphi_{ij}(x) \overline{\varphi'_{j'j'}(x)}] \tag{from Lemma~\ref{lem:trace-of-fourier-coefficient}}\\
    & =
    \sum_{\varphi,\varphi'}d_{\varphi}\sum_{i,j} \widehat{f}(\varphi)_{ij} \sum_{i'}\overline{\widehat{g}(\varphi')_{i'i'}} \sum_{j'} \langle \varphi_{ij}, \varphi'_{j'j'}\rangle \\
    & =
    \sum_{\varphi}\sum_{i} \widehat{f}(\varphi)_{ii} \sum_{i'}\overline{\widehat{g}(\varphi)_{i'i'}}  \\
    & =
    \sum_{\varphi} \tr(\widehat{f}(\varphi))\overline{\tr(\widehat{g}(\varphi))}.
    \qedhere
  \end{align*}
\end{proof}

\begin{proof}[Proof of Lemma~\ref{lem:fourier-expansion-of-class-function}]
  \begin{align*}
    & \| f\|_{2}^{2} = \sum_{\varphi} \Bigl| \tr(\widehat{f}(\varphi))\Bigr|^{2} \quad \tag{from Lemma~\ref{lem:distance-to-class-functions}}\\
    = & \sum_{\varphi} \Bigl| \sum_{i\in [d_{\varphi}]} \widehat{f}(\varphi)_{ii}\Bigr|^{2}
    \leq \sum_{\varphi} d_{\varphi} \sum_{i\in [d_{\varphi}]} \Bigl| \widehat{f}(\varphi)_{ii}\Bigr|^{2}  \tag{by Cauchy-Schwarz}\\
    \leq & \sum_{\varphi} d_{\varphi} \sum_{i, j\in [d_{\varphi}]} \left| \widehat{f}(\varphi)_{ij}\right|^{2}
    = \| f\|_{2}^{2}. \tag{by Parseval's identity}
  \end{align*}
  Therefore, the equality holds for both inequalities in the formula above.
  In particular, $\widehat{f}(\varphi)$ is proportional to the identity matrix $cI_{d_\varphi}$ for some $c \in \bbC$.
  By Lemma~\ref{lem:trace-of-fourier-coefficient},
  $\langle f, \chi_{\varphi}\rangle = \tr(\widehat{f}(\varphi)) = cd_\varphi$.
  Hence, $c = \langle f, \chi_\varphi \rangle / d_\varphi$, and we have the lemma.
\end{proof}

\subsection{Introduction to the Haar measure}\label{app:haar-measure}
In this section, we introduce Haar measure briefly.
See, e.g., a textbook~\cite{Alexander2008LieGroup} for more details.


A \emph{topological group} is a group equipped with a topology and whose group operations are continuous in its topology.
A finite group is a topological group if it is endowed with a discrete topology.
Any subgroup of $M_{d}(\bbC)$ is a topological group, in which we identify $M_{d}(\bbC)$ with $\bbC^{d^{2}}$ and introduce the topology induced by $\bbC^{d^{2}}$.

We call a measure $\mu$ on a topological group $G$ \emph{left invariant} (resp., \emph{right invariant}) if $\mu(xS) = \mu(S)$ (resp., $\mu(S) = \mu(Sx)$) for any $x\in G$ and Borel set $S \subseteq G$.
Similarly, We call $\mu$ \emph{invariant under taking inverse} if $\mu(S^{-1}) = \mu(S)$ for any Borel set $S \subseteq G$ where $S^{-1} = \{x^{-1} \mid x\in S\}$.

For any compact topological group $G$, there exists a measure on $G$ which is left invariant, right invariant, and invariant under taking inverse.
Such a measure is unique up to scalar multiplication and called the Haar measure on $G$.
For example, the Haar measure $\mu$ of a finite group $G$ is (a scalar multiplication of) the counting measure, that is, $\mu(S) = |S|/|G|$ for subset $S \subseteq G$.

We regard $U(d)$ as a closed subgroup of $M_{d}(\bbC)$ and regard it as a compact topological group.
Hence, the Haar measure of $U(d)$ exists.

\section{Conjugate Invariance}\label{sec:conjugate}

In this section, we first show that conjugate-invariance is constant-query testable.
\begin{theorem}\label{the:conjugate-invariance}
  Conjugate invariance is one-sided error testable with $O(1/\epsilon^2)$ queries.
\end{theorem}

Then, we show the following lemma, which simplifies testing properties that imply conjugate invariance.
\begin{lemma}\label{lem:reduced-to-class-function}
  Let $P$ be a property such that every $f$ satisfying $P$ is a class function.
  Suppose that there is a tester $\caA$ for $P$ with query complexity $q(\epsilon)$ if the input is restricted to be a class function.
  Then there is a tester $\caA'$ for $P$ with query complexity $O(1/\epsilon^2 + q(\epsilon/2) \log q(\epsilon/2))$.
  Moreover, if $\caA$ is a one-sided error tester, then $\caA'$ is also a one-sided error tester.
\end{lemma}


\subsection{Proof of Theorem~\ref{the:conjugate-invariance}}\label{subsec:testing-conjugate-invariance}
\begin{algorithm}[t!]
  \caption{(Tester for conjugate invariance)}
  \label{alg:conjugate-invariance}
  \begin{algorithmic}[1]
  \FOR{$s := O(1/\epsilon^2)$ times}
  \STATE Sample $x$ and $y \in G$ uniformly at random.
  \STATE \textbf{if} $f(x) \neq f(yxy^{-1})$ \textbf{then} reject. \label{line:conjugate-invariance-rejection}
  \ENDFOR
  \STATE Accept.
  \end{algorithmic}
\end{algorithm}
Our algorithm for testing conjugate invariance is described in Algorithm~\ref{alg:conjugate-invariance}.
It is easy to see that the query complexity of Algorithm~\ref{alg:conjugate-invariance} is $O(1/\epsilon^2)$ and the tester always accepts when $f$ is conjugate invariant.
Thus, it suffices to show that Algorithm~\ref{alg:conjugate-invariance} rejects with probability at least $2/3$ when $f$ is $\epsilon$-far from class functions.

It is well known that conjugacy classes of $G$ form a partition of $G$.
We define $G^\sharp$ as the set of conjugacy classes of $G$.
Also, for an element $x \in G$, we define $x^\sharp$ as the unique conjugacy class $x$ belongs to.

For $y\in x^\sharp$, we define $N_{x, y} = \{z\in G \mid zxz^{-1} = y \}$.
Since we have $N_{x, y} \cap N_{x, y'} = \emptyset$
for $y, y' \in x^\sharp$ with $y \neq y'$,
it holds that $G=\bigsqcup_{y\in x^\sharp} N_{x, y}$.
Therefore, the following lemma guarantees that uniform sampling from a conjugacy class is executed by uniformly sampling from the whole group.
\begin{lemma}\label{lem:sampling-from-conjugacy-class}
  The number of elements in $N_{x, y}$ depends only on the conjugacy class to which $y$ belongs.
\end{lemma}
\begin{proof}
  For $y, y'\in x^\sharp$, fix $z_{0}, z'_{0}\in G$ so that $z_{0}xz_{0}^{-1} = y$ and $z'_{0}xz'_{0}{}^{-1} = y'$ hold.
  We construct mappings $\Phi_{y, y'} : N_{x, y}\to N_{x, y'}$ by $z \mapsto z'_{0}z^{-1}z_{0}$ and $\Phi_{y', y} : N_{x, y'}\to N_{x, y}$ by $z \mapsto z_{0}z^{-1}z'_{0}$.
  By a direct calculation, we can check that $ \Phi_{y, y'}\circ \Phi_{y', y} = \id_{N_{x, y}}$, and $\Phi_{y', y}\circ \Phi_{y, y'} = \id_{N_{x, y'}}$.
  Therefore $|N_{x, y}| = |N_{x, y'}|$ holds.
\end{proof}

Fix a function $f:G \to \bbD$.
For a conjugacy class $C \in G^\sharp$ and $z \in \bbC$,
define $p_C(z) := \sharp \{x \in C \mid f(x) = z\} / |C|$ as the probability that $f(x) = z$ if we sample $x \in C$ uniformly at random.
We define $p_C := \max_{z \in \bbC}p_C(z)$ and $z_C := \arg \max_{z \in \bbC}p_C(z)$.
Then, we define $\widetilde{f}$ as $\widetilde{f}(x) = z_{x^\sharp}$.
Note that $\widetilde{f}$ is a class function such that $\widetilde{f}(x) \in \bbD$ for any $x\in G$.



\begin{lemma}\label{lem:probability-of-conjugate-invariance}
  \begin{align*}
    \Pr_{x,y\in G}[f(x) \neq f(yxy^{-1})] \geq \dist(f,\widetilde{f})^2.
  \end{align*}
\end{lemma}
\begin{proof}
  Since $|x - y| \leq 2$ for any $x, y \in \bbD$,
  we have
  \[
    \dist(f,\widetilde{f})^2
    = \frac{1}{4}\E_x |f(x) - \widetilde{f}(x)|^2
    \leq \Pr_x[f(x) \neq \widetilde{f}(x)]
    = \frac{1}{|G|}\sum_{C \in G^\sharp} |C|(1-p_C).
  \]

  By Lemma~\ref{lem:sampling-from-conjugacy-class},
  if we fix $x \in G$ and sample $y \in G$ uniformly at random,
  then $yxy^{-1}$ forms a uniform distribution over elements in $x^\sharp$.
  Thus,
  \begin{align*}
    & \Pr_{x,y\in G}[f(x) \neq f(yxy^{-1})]
    \geq
    \frac{1}{|G|} \sum_{C \in G^\sharp}|C| \int p_C(z)(1-p_C(z)) \mathrm{d}z \\
    & \geq
    \frac{1}{|G|} \sum_{C \in G^\sharp}|C| \int p_C(z)(1-p_C) \mathrm{d}z
    =
    \frac{1}{|G|} \sum_{C \in G^\sharp}|C|(1-p_C)
    \geq
    \dist(f,\widetilde{f})^2.
  \end{align*}
\end{proof}
\begin{lemma}\label{lem:conjugate-invariance-soundness}
  If $f$ is $\epsilon$-far from being conjugate invariant,
  then Algorithm~\ref{alg:conjugate-invariance} rejects with probability at least $2/3$.
\end{lemma}
\begin{proof}
  Since $\widetilde{f}:G \to \bbD$ is a class function, we have $\dist(f,\widetilde{f}) \geq \epsilon$.
  Hence, the probability we reject at Line~\ref{line:conjugate-invariance-rejection} in each trial is at least $\epsilon^2$ by Lemma~\ref{lem:probability-of-conjugate-invariance}.
  Hence the tester rejects with probability $2/3$ by choosing the hidden constant in $s$ large enough.
\end{proof}
We establish Theorem~\ref{the:conjugate-invariance} by Lemma~\ref{lem:conjugate-invariance-soundness}.

\subsection{Proof of Lemma~\ref{lem:reduced-to-class-function}}\label{subsec:self-correction-to-class-function}
The following lemma shows that we can obtain a query access to a class function that is close to $f$.
\begin{lemma}\label{lem:self-correction-to-class-function}
  Let $f:G \to \bbD$ be a function that is $\epsilon$-close to a class function.
  There exists a class function $f':G \to \bbD$ with the following property.
  \begin{itemize}
  \setlength{\itemsep}{0pt}
  \item For any $x \in G$,
  with $O(\log 1/\delta)$ queries to $f$,
  we can correctly compute $f'(x)$ or find a witness that $f$ is not a class function with probability at least $1-\delta$.
  Moreover, if $f$ itself is a class function, then we can always compute $f'(x)$ correctly.
  \item $\dist(f',f) \leq 3\epsilon$. In particular, $f' = f$ when $f$ itself is a class function.
  \end{itemize}
\end{lemma}
\begin{proof}
  For a conjugacy class $C$, let $z^*_C \in \bbD$ be the unique value that minimizes $\sum_{x \in C}|f(x) - z^*_C|^2$.
  We define $f^*:G \to \bbD$ as $f^*(x) = z^*_{x^\sharp}$.
  Note that $f^*$ is the class function closest to $f$.

  We define $f':G \to \bbD$ as follows:
  \[
    f'(x) = \begin{cases}
    z^*_{x^\sharp} & \text{if }p_{x^\sharp} \leq \frac 12,\\
    z_{x^\sharp} & \mbox{otherwise}.
    \end{cases}
  \]

  We first show the first claim.
  Our algorithm for computing $f'(x)$ is as follows.
  Given $x \in G$,
  we pick $y_1,\ldots,y_s \in G$ for $s := O(\log 1/\delta)$ uniformly at random,
  and compute $f(y_ixy_i^{-1})$ for each $i \in [s]$.
  If $f(y_ixy_i^{-1}) \neq f(y_jxy_j^{-1})$ for some $i \neq j$,
  then we reject $f$ and output the pair as the witness that $f$ is not a class function.
  If all of them are the same, we output the value as $f'(x)$.

  Now we analyze the correctness of the algorithm above.
  If $p_{x^\sharp} \leq 1/2$,
  then with probability at least $1-\delta$,
  we have $f(y_i x y_i^{-1}) \neq f(y_j x y_j^{-1})$ for some $i \neq j$,
  and we reject.
  If $p_{x^{\sharp}} > 1/2$,
  with probability at least $1-\delta$,
  the majority of $\{f(y_ixy_i^{-1})\}_{i \in [s]}$ is equal to $z_C$.
  Hence,
  with probability at least $1-\delta$,
  either we output $z_C$ or reject.

  Moreover, if $f$ itself is a class function, then we have $f'(x) = f(x)$ for any $x \in G$, and our algorithm always outputs $f(x)$ as $f'(x)$.

  We turn to the second claim.
  For two functions $g,h:G\to \bbC$ and a conjugacy class $C$,
  define $\dist_C(g,h) := \sqrt{\sum_{x \in C}|g(x) - h(x)|^2}$.
  We will show that, for each conjugacy class $C \in G^\sharp$,
  $\dist_C(f,f') \leq 3\dist_C(f,f^*)$,
  which implies $\dist(f,f') \leq 3\epsilon$.

  If $C$ satisfies $p_C \leq 1/2$, we have nothing to show.
  Thus suppose $p_C > 1/2$.
  Then, we have
  \begin{align*}
    \dist_C(f,f') \leq \dist_C(f,f^*) + \dist_C(f^*, f')
    = \sqrt{\sum_{x \in C}|f(x) - f^*(x)|^{2}} + \sqrt{|C||z^*_C - z_C|^{2}}.
  \end{align*}
  Since a $p_C$-fraction of values has moved from $z_C$ to $z^*_C$ when constructing $f^*$ from $f$,
  we have $p_C|C||z^*_C - z_C|^{2} \leq \sum_{x \in C}|f(x) - f^{*}(x)|^{2}$.
  By $p_C > 1/2$, we have $|C||z^*_C - z_C|^{2} \leq 2\sum_{x \in C}|f(x) - f^{*}(x)|^{2}$.
  Combining this with the previous inequality, we have $\dist_C(f,f') \leq (1+\sqrt{2})\dist_C(f,f^*) \leq 3\dist_C(f,f^*)$.
\end{proof}

\begin{proof}[Proof of Lemma~\ref{lem:reduced-to-class-function}]
  We first apply the $\epsilon/6$-tester for conjugate invariance (Algorithm~\ref{alg:conjugate-invariance}).
  If the tester rejects, we immediately reject $f$ as it implies that $f$ does not satisfy $P$.
  Otherwise, using Lemma~\ref{lem:self-correction-to-class-function},
  we construct a query access to a class function $f'$ with $\delta = O(1/q(\epsilon/2))$.
  Then we apply the tester $\caA$ to $f'$ with the error parameter $\epsilon/2$.
  The query complexity is clearly as stated.

  Suppose that $f$ satisfies the property $P$.
  Then, we never reject when testing conjugate invariance.
  Also $f'(x) = f(x)$ holds for every $x \in G$ and it follows that $f'$ satisfies the property $P$.
  Hence, the tester $\caA$ accepts $f'$ with probability at least $2/3$.
  Moreover if $\caA$ is a one-sided error tester, then $\caA$ accepts $f'$ with probability one.

  Suppose that $f$ is $\epsilon$-far from the property $P$.
  If $f$ is $\epsilon/6$-far from conjugate invariance, then we reject $f$ with probability at least $2/3$.
  Thus assume that $f$ is $\epsilon/6$-close to conjugate invariance.
  In this case $f'$ is a class function that is $\epsilon/2$-close to $f$.
  Hence $f'$ is still $\epsilon/2$-far from the property $P$.
  Then the tester $\caA$ on $f'$ should reject with probability $2/3$.
\end{proof}

\section{Testing Homomorphism}\label{sec:homo}

In this section, we show the following:
\begin{theorem}\label{the:homomorphism}
  Homomorphism is one-sided error testable with $O(1/\epsilon^2 \log (1/\epsilon))$ queries.
\end{theorem}
We note that, if $f:G \to \bbC$ is a homomorphism, then it is a one-dimensional representation and hence an irreducible representation.
First we observe that homomorphism implies conjugate invariance.
\begin{lemma}\label{lem:homomorphism->class-function}
  If $f:G \to \bbC$ is a homomorphism,
  then $f$ is conjugate invariant.
\end{lemma}
\begin{proof}
  Since $f$ is a homomorphism,
  we have for any $x, y \in G$, $f(yxy^{-1}) = f(y)f(x)f(y^{-1}) = f(y)f(y^{-1})f(x) = f(yy^{-1})f(x) = f(1)f(x)$.
  By setting $y = 1$, we have $f(x) = f(1)f(x)$, which means $f(x) = 0$ or $f(1) = 1$.

  If $f(x) = 0$ for all $x \in G$, then $f$ is clearly conjugate invariant.
  If $f(x) \neq 0$ for some $x \in G$, then $f(1) = 1$.
  In this case, we have $f(yxy^{-1}) = f(1)f(x) = f(x)$ and $f$ is again conjugate invariant.
\end{proof}

\begin{algorithm}[!t]
  \caption{(Tester for homomorphism)}
  \label{alg:homomorphism}
  \begin{algorithmic}[1]
  \REQUIRE{A class function $f:G \to \bbD$.}
  \FOR{$s = O(1/\epsilon^2)$ times}
  \STATE Sample $x,y \in G$ uniformly at random.
  \STATE \textbf{if} $f(x)f(y) \neq f(xy)$ \textbf{then} reject.
  \ENDFOR
  \STATE Accept.
  \end{algorithmic}
\end{algorithm}

From Lemmas~\ref{lem:reduced-to-class-function} and~\ref{lem:homomorphism->class-function},
to test homomorphism, it suffices to show that homomorphism is one-sided error testable with $O(1/\epsilon^2)$ queries when the input function is a class function.
Our tester is given in Algorithm~\ref{alg:homomorphism}.
It is clear that the query complexity is $O(1/\epsilon^2)$.
We next see that Algorithm~\ref{alg:homomorphism} always accepts homomorphisms:
\begin{lemma}
  If a class function $f:G \to \bbD$ is a homomorphism, then Algorithm~\ref{alg:homomorphism} always accepts.
\end{lemma}
\begin{proof}

  We always accept because $f(x)f(y) = f(xy)$ for any $x,y \in G$
\end{proof}

Now we turn to the case that $f$ is $\epsilon$-far from homomorphisms.
To show that $\Pr[f(x)f(y) \neq f(xy)]$ is much smaller than $1$,
we analyze the term $f(x)f(y)\overline{f(xy)}$.
\begin{lemma}\label{lem:estimate-of-homomorphism}
  For any function $f:G \to \bbC$, we have
  \[
    \E_{x,y}[f(x)f(y)\overline{f(xy)}] =  \sum_{\varphi}d_\varphi \sum_{i,j,k \in [d_\varphi]}\widehat{f}(\varphi)_{ij}\widehat{f}(\varphi)_{jk}\overline{\widehat{f}(\varphi)_{ik}}.
  \]
\end{lemma}
\begin{proof}
  The left hand side is equal to
  \begin{align}
    \sum_{\varphi,\varphi',\varphi''}d_\varphi d_{\varphi'} d_{\varphi''}\sum_{i,j \in [d_\varphi]}\sum_{i',j' \in [d_{\varphi'}]}\sum_{i'',j'' \in [d_{\varphi''}]}\widehat{f}(\varphi)_{ij}\widehat{f}(\varphi')_{i'j'}\overline{\widehat{f}(\varphi'')_{i''j''}} \E_{x,y}[ \varphi_{ij}(x) \varphi'_{i'j'}(y) \overline{\varphi''_{i''j''}(xy)}]. \label{eq:cubic-expansion}
  \end{align}
  Now we analyze the expectation in~\eqref{eq:cubic-expansion}.
  \begin{align*}
    & \E_{x,y}[ \varphi_{ij}(x) \varphi'_{i'j'}(y) \overline{\varphi''_{i''j''}(xy)}]
    =
    \E_{x,y}\Bigl[ \varphi_{ij}(x) \varphi'_{i'j'}(y) \sum_{k'' \in [d_{\varphi''}]}\overline{\varphi''_{i''k''}(x)}\overline{\varphi''_{k''j''}(y)}\Bigr] \\
    & =
    \sum_{k'' \in [d_{\varphi''}]}\langle \varphi_{ij},\varphi''_{i''k''}\rangle \langle \varphi'_{i'j'} \varphi''_{k''j''} \rangle
    =
    \begin{cases}
    \frac{1}{d_\varphi^2} & \mbox{if } \varphi = \varphi' = \varphi'', i = i'', j = i' = k'',\mbox{ and }j' = j'',\\
    0 & \mbox{otherwise}.
    \end{cases}
  \end{align*}
  Hence $\eqref{eq:cubic-expansion} = \sum_{\varphi}d_\varphi \sum_{i,j,j' \in [d_\varphi]}\widehat{f}(\varphi)_{ij}\widehat{f}(\varphi)_{jj'}\overline{\widehat{f}(\varphi)_{ij'}}$.
\end{proof}
\begin{corollary}\label{cor:estimate-of-homomorphism}
  For any class function $f:G \to \bbC$, we have
  \[
    \E_{x,y}[f(x)f(y)\overline{f(xy)}] = \sum_{\varphi}d_\varphi \sum_{i \in [d_\varphi]}\widehat{f}(\varphi)_{ii}|\widehat{f}(\varphi)_{ii}|^2.
  \]
\end{corollary}
\begin{proof}
  If $f$ is a class function, then $\widehat{f}(\varphi)_{ij} = 0$ for any $\varphi \in \widehat{G}$ and $i \neq j \in [d_\varphi]$ by Lemma~\ref{lem:fourier-expansion-of-class-function}.
  Hence, we have the corollary from Lemma~\ref{lem:estimate-of-homomorphism}.
\end{proof}

The following lemma completes the proof of Theorem~\ref{the:homomorphism}.
\begin{lemma}
  If a class function $f:G \to \bbD$ is $\epsilon$-far from homomorphism,
  then Algorithm~\ref{alg:homomorphism} rejects with probability at least $2/3$.
\end{lemma}
\begin{proof}

  From Corollary~\ref{cor:estimate-of-homomorphism}, we have
  \begin{align*}
    & \Re \E_{x,y}[f(x)f(y)\overline{f(x+y)}]
    =
    \Re \sum_{\varphi}d_\varphi \sum_{i \in [d_\varphi]}\widehat{f}(\varphi)_{ii}|\widehat{f}(\varphi)_{ii}|^2 \\
    & \leq
    \max_{\varphi \in \widehat{G}, i \in [d_\varphi]} \Re\widehat{f}(\varphi)_{ii}\cdot  \sum_{\varphi}d_\varphi \sum_{i \in [d_\varphi]}|\widehat{f}(\varphi)_{ii}|^2 \\
    & =
    \max_{\varphi \in \widehat{G}, i \in [d_\varphi]} \Re\widehat{f}(\varphi)_{ii} \cdot \|f\|_2^2
    \leq
    \max_{\varphi \in \widehat{G}, i \in [d_\varphi]} \Re \widehat{f}(\varphi)_{ii}.
  \end{align*}
  For any $\varphi$ with dimension more than one,
  $|\widehat{f}(\varphi)_{ii}| \leq 1/2$ and hence $\Re \widehat{f}(\varphi)_{ii} \leq 1/2$ (see Lemma~\ref{lem:fourier-expansion-of-class-function}).
  Now consider a one-dimensional irreducible representation $\varphi$.
  Since $f$ is $\epsilon$-far from homomorphism, we have
  \[
    \epsilon \leq \dist(f,\varphi)
    = \frac{1}{2}\sqrt{\|f\|_2^2 + \|\varphi\|^2_2 - 2\Re \langle f, \varphi\rangle }
    \leq \frac{1}{2}\sqrt{2 - 2\Re\widehat{f}(\varphi)}.
  \]
  Note that $\|f\|_2^2 \leq 1$ and $\|\varphi\|_2^2 = 1$ as $\varphi$ is a (non-zero) homomorphism.
  Hence, $\Re\widehat{f}(\varphi) \leq 1 - 2\epsilon^2$.

  We have shown that $\Re\E_{x,y}[f(x)f(y)\overline{f(x+y)}] \leq 1-2\epsilon^2$.
  Since $|f(x)f(y)\overline{f(x+y)}| \leq 1$,
  at least an $\Omega(\epsilon^2)$-fraction of pairs $(x, y)$ satisfy $f(x)f(y) \overline{f(x+y)} \neq 1$.
  Hence we have $\Pr_{x,y}[f(x)f(y) = f(x+y)] \leq 1-\Omega(\epsilon^2)$.
  By choosing the hidden constant in $s$ large enough,
  we reject with probability at least $2/3$.
\end{proof}

\section{Testing the Property of Being Proportional to an Irreducible Character}\label{sec:nrc}

In this section, we show the following:
\begin{theorem}
  The property of being proportional to an irreducible character is testable with $O(1/\epsilon^8 \log^2(1/\epsilon))$ queries.
\end{theorem}
As any character is a class function,
by Lemma~\ref{lem:reduced-to-class-function},
it suffices to give a tester with query complexity $O(1/\epsilon^8 \log(1/\epsilon))$ that works when the input function is a class function.
The following fact is crucial for our algorithm.
\begin{lemma}[\cite{Weyl:1927dn}]\label{lem:characterization-of-character}
  For a function $f:G \to \bbC$, the following are equivalent.
  \begin{enumerate}
  \itemsep=0pt
  \item $f(x) = f(1) \widetilde{\chi}_\varphi(x) $ for some irreducible representation $\varphi$,  where $\widetilde{\chi}_\varphi = \chi_\varphi / d_\varphi$.
  \item $f(x)f(y) = f(1)\E_{z \in G}[ f(yzxz^{-1})]$ for any $x, y \in G$.
  \end{enumerate}
\end{lemma}
As we can freely change the value of $f(1)$ by multiplying a constant,
the second condition is a necesary and sufficient condition of being proportional to an irreducible character.

The most simple test based on Lemma~\ref{lem:characterization-of-character} is checking whether $f(x)f(y) \approx f(1)\E_{z}[ f(yzxz^{-1})]$ (by estimating the latter by sampling $z \in G$ a constant number of times).
However, we were unable to handle the term $\E_{x,y,z}[f(x)f(y)\overline{ f(yzxz^{-1})}\overline{f(1)}]$ that naturally arises when analyzing this test.
Instead, we estimate $|f(x)f(y) - f(1)\E_z[f(xyzxz^{-1})]|$ and check whether it is small.
The detail is given in Algorithm~\ref{alg:normalized-character}.
It is clear that the query complexity of Algorithm~\ref{alg:normalized-character} is $O(1/\epsilon^8 \log 1/\epsilon)$.

\begin{algorithm}[!t]
  \caption{(Tester for being proportional to an irreducible character)}
  \label{alg:normalized-character}
  \begin{algorithmic}[1]
  \REQUIRE{A function $f:G \to \bbD$.}
  \STATE Let $e_{\ref{line:estimate-f_2^2}}$ be the estimation to $\|f\|_2^2$ obtained by applying Lemma~\ref{lem:estimate} with the error parameter $\epsilon^2/100$ and the confidence parameter $1/100$. \label{line:estimate-f_2^2}
  \STATE \textbf{if} $e_{\ref{line:estimate-f_2^2}} < \epsilon^2 / 2$ \textbf{then} accept.
  \FOR{each $i = 1$ to $s := O(1/\epsilon^4)$}
    \STATE Sample $x,y\in G$ uniformly at random.
    \STATE Let $e^i_{\ref{line:estimate-f(yzxz)}}$ be the estimation to $\E_z [f(yzxz^{-1})]$ obtained by applying Lemma~\ref{lem:estimate} with the error parameter $\epsilon^2/10$ and the confidence parameter $1/100s$.  \label{line:estimate-f(yzxz)}
    \STATE Let $e^i_{\ref{line:estimate-fxfy-e}} = |f(x)f(y) - f(1)e^i_{\ref{line:estimate-f(yzxz)}}|^2$. \label{line:estimate-fxfy-e}
    \STATE \textbf{if} $e^i_{\ref{line:estimate-fxfy-e}} > \epsilon^4/100$ \textbf{then} reject. \label{line:alg3-reject}
  \ENDFOR
  \STATE accept.
  \end{algorithmic}
\end{algorithm}

\begin{lemma}
  If a function $f:G \to \bbC$ is proportional to an irreducible character, then Algorithm~\ref{alg:normalized-character} accepts with probability at least $2/3$.
\end{lemma}
\begin{proof}
  By the union bound, all the estimations succeed with probability at least $2/3$.
  Below we assume this indeed happens.

  Recall that $f(x)f(y) - f(1)\E_z[f(yzxz^{-1})] = 0$ for any $x, y \in G$ by Lemma~\ref{lem:characterization-of-character}.
  Then for each $i$,
  $e^i_{\ref{line:estimate-fxfy-e}} = |f(x)f(y) - f(1)e^i_{\ref{line:estimate-f(yzxz)}}|^2 = |f(1)\E_z[f(yzxz^{-1})] - f(1)e^i_{\ref{line:estimate-f(yzxz)}}|^2 \leq \epsilon^4/100$ holds for every $i \in [s]$.
  Hence, we accept with probability at least $2/3$.
\end{proof}

Now we turn to the case that $f$ is $\epsilon$-far from being proportional to an irreducible character.
We need the following auxiliary lemma.
\begin{lemma}\label{lem:sum-of-squared-difference}
  For any function $f: G \to \bbC$, we have
  \[
    \E_{x, y}\Bigl[ \bigl| f(x)f(y) -f(1)\E_{z}[f(yzxz^{-1})] \bigr|^{2}  \Bigr] \geq  \|f\|_{2}^{2} \min_{\varphi} \|f - f(1)\widetilde{\chi}_{\varphi} \|_{2}^{2}.
  \]
\end{lemma}
\begin{proof}
  We have
  \begin{align*}
    & \E_{z}\Bigl[ f(yzxz^{-1}) \Bigr]
    = \E_{z}\Bigl[ \sum_{\varphi} d_{\varphi} \sum_{i, j\in [d_{\varphi}]} \widehat{f}(\varphi)_{ij} \varphi_{ij}(yzxz^{-1}) \Bigr] \\
    = & \E_{z}\Bigl[ \sum_{\varphi} d_{\varphi} \sum_{i, j\in [d_{\varphi}]} \widehat{f}(\varphi)_{ij} \sum_{k, l, m\in [d_{\varphi}]} \varphi_{ik}(y)\varphi_{kl}(z)\varphi_{lm}(x)\varphi_{mj}(z^{-1}) \Bigr]\\
    =& \sum_{\varphi} d_{\varphi} \sum_{i, j\in [d_{\varphi}]} \widehat{f}(\varphi)_{ij} \sum_{k, l, m\in [d_{\varphi}]} \varphi_{ik}(y)\varphi_{lm}(x)\E_{z}\Bigl[\varphi_{kl}(z)\overline{\varphi_{jm}(z)} \Bigr] \\
    = & \sum_{\varphi} d_{\varphi} \sum_{i, j\in [d_{\varphi}]} \widehat{f}(\varphi)_{ij} \sum_{k, l, m\in [d_{\varphi}]} \varphi_{ik}(y)\varphi_{lm}(x) \frac{\delta_{kj}\delta_{lm}}{d_{\varphi}}\\
    =& \sum_{\varphi} d_{\varphi} \sum_{i, j\in [d_{\varphi}]} \widehat{f}(\varphi)_{ij} \varphi_{ij}(y) \sum_{k\in [d_{\varphi}]} \frac{\varphi_{kk}(x)}{d_{\varphi}}
    = \sum_{\varphi} d_{\varphi} \widetilde{\chi}_{\varphi}(x) \sum_{i, j\in [d_{\varphi}]} \widehat{f}(\varphi)_{ij} \varphi_{ij}(y)
  \end{align*}
In the third equality, we used the fact that $\varphi_{mj}(z^{-1}) = \overline{\varphi_{jm}(z)}$.
This follows from the fact that $\varphi(z^{-1})\varphi(z) = \varphi(1) = I$ and $\varphi(z)$ is unitary.

  Therefore,
  \begin{align*}
    & f(x)f(y) - f(1)\E_{z}[f(yzxz^{-1})]
    = f(x)\sum_{\varphi} d_{\varphi} \sum_{i, j\in [d_{\varphi}]} \widehat{f}(\varphi)_{ij} \varphi_{ij}(y) - f(1)\sum_{\varphi} d_{\varphi} \widetilde{\chi}_{\varphi} \sum_{i, j\in [d_{\varphi}]} \widehat{f}(\varphi)_{ij} \varphi_{ij}(y) \\
    = & \sum_{\varphi} d_{\varphi} ( f(x) - f(1)\widetilde{\chi}_{\varphi}(x)) \sum_{i, j\in [d_{\varphi}]} \widehat{f}(\varphi)_{ij} \varphi_{ij}(y)
  \end{align*}
  It follows that
  \begin{align*}
    & \E_{x, y}\Bigl[ \bigl| f(x)f(y) -  f(1)\E_{z}[f(yzxz^{-1})] \bigr|^{2}  \Bigr]
    = \E_{x, y}\Bigl[ \bigl|\sum_{\varphi} d_{\varphi} ( f(x) - f(1)\widetilde{\chi}_{\varphi}(x)) \sum_{i, j\in [d_{\varphi}]} \widehat{f}(\varphi)_{ij} \varphi(y)_{ij} \bigr|^{2} \Bigr]\\
    = & \sum_{\varphi, \varphi'} d_{\varphi}d_{\varphi'} \E_{x}\bigl[ (f(x) - f(1)\widetilde{\chi}_{\varphi}(x))\overline{(f(x) - f(1)\widetilde{\chi}_{\varphi'}(x))}\bigr] \sum_{i, j\in [d_{\varphi}]} \sum_{i', j'\in [d_{\varphi'}]} \widehat{f}(\varphi)_{ij} \overline{\widehat{f}(\varphi')_{i'j'}} \E_{y}\bigl[\varphi_{ij}(y)\overline{\varphi'_{i'j'}(y)} \bigr]\\
    = & \sum_{\varphi} d_{\varphi} \| f - f(1)\widetilde{\chi}_{\varphi}\|_{2}^{2} \sum_{i, j\in [d_{\varphi}]} \left|\widehat{f}(\varphi)_{ij}\right|^{2}
    \geq   \min_{\varphi}\| f - f(1)\widetilde{\chi}_{\varphi}\|_{2}^{2} \sum_{\varphi} d_{\varphi}  \sum_{i, j\in [d_{\varphi}]} \left|\widehat{f}(\varphi)_{ij}\right|^{2}\\
    \geq & \|f\|_{2}^{2} \min_{\varphi}\| f - f(1)\widetilde{\chi}_{\varphi}\|_{2}^{2}
  \end{align*}
  In the third equality, we used the fact that $\E_{y}\bigl[\varphi_{ij}(y)\overline{\varphi'_{i'j'}(y)} \bigr]$ is equal to $1/d_\varphi$ if $\varphi = \varphi'$, $i = i'$, and $j = j'$, and is equal to zero otherwise.
\end{proof}

\begin{lemma}
  If a function $f:G \to \bbC$ with $f(1) = 1$ is $\epsilon$-far from being proportional to an irreducible character,
  then Algorithm~\ref{alg:normalized-character} rejects with probability at least $2/3$.
\end{lemma}
\begin{proof}
  By the union bound,
  with probability at least $5/6$,
  all the estimations succeed.
  Below we assume it indeed happens.

  From Lemma~\ref{lem:sum-of-squared-difference},
  when $f$ is $\epsilon$-far,
  the expectation of $e^i_{\ref{line:estimate-fxfy-e}}$ is at least
  \[
    \E_{x, y}\Bigl[ \bigl| f(x)f(y) -f(1)e_{\ref{line:estimate-f(yzxz)}}^i \bigr|^{2}  \Bigr]
    \geq  (1-\frac{2\epsilon^2}{10})\|f\|_{2}^{2} \min_{\varphi} \|f - f(1)\widetilde{\chi}_{\varphi} \|_{2}^{2} - \frac{\epsilon^4}{100}
    \geq (1-\frac{\epsilon^2}{5}) \frac{\epsilon^2}{4} \cdot \|f\|_{2}^{2} - \frac{\epsilon^4}{100} \geq \frac{\epsilon^4}{25}.
  \]
  We also note that
  $e^i_{\ref{line:estimate-fxfy-e}} \leq (1 + \epsilon^2/100)^2 \leq 2$.
  Let $p = \Pr_{x,y}[e^i_{\ref{line:estimate-fxfy-e}} > \epsilon^4/100]$.
  Then, we have $2 \cdot p + \epsilon^4/100\cdot (1-p) \geq \epsilon^4/25$, and it follows that $p \geq \epsilon^4/100$.
  By choosing the hidden constant in $s$ large enough, we reject $f$ with probability at least $2/3$.
\end{proof}

\section{Testing Unitary Equivalence}\label{sec:unitary}

In this section, we prove the following:
\begin{theorem}\label{the:unitary}
  The unitary equivalence to $g:G \to \bbD(d)$ is testable with $\bigl(d^{3/2}/\epsilon\bigr)^{O(d^2)}$ queries.
\end{theorem}

Our algorithm is described in Algorithm~\ref{alg:equivalence}.
We use the Haar measure on $U(d)$ to sample unitary matrices.
We do not need the detailed definition of the Haar measure, and we only have to understand that it defines a probability distribution on $U(d)$.
See Section~\ref{app:haar-measure} for a brief introduction to the Haar measure.

The basic idea of our algorithm and analysis is the following.
Suppose that functions $f,g:G \to \bbD(d)$ are unitary equivalent, that is, $f = U_0 g U_0^*$ for some unitary matrix $U_0 \in U(d)$.
Then, by sampling a sufficient number of unitary matrices from the Haar measure, we get a unitary matrix $U$ that is sufficiently close to $U_0$ in the sense that the Frobenius norm of $U - U_0$ is small (Lemma~\ref{lem:unitary-epsilon-net}).
Then, we can show that the Frobenius norm of $f(x) - U g(x)U^*$ is also small for any $x \in G$ (Lemma~\ref{lem:U-V->UAU-VAV}).
On the other hand, if $f$ and $g$ are $\epsilon$-far from being unitary equivalent, then the average Frobenius norm of $f(x) - Ug(x)U^*$ over $x \in G$ is large for any unitary matrix $U$.
Hence, by checking whether there is a unitary matrix $U$ (in the sample) such that the average Frobenius norm is small, we can distinguish the case that $f$ and $g$ are unitary equivalent from the case that $f$ and $g$ are $\epsilon$-far from being unitary equivalent.

\begin{algorithm}[!t]
  \caption{(Tester for unitary equivalence)}
  \label{alg:equivalence}
  \begin{algorithmic}[1]
    \REQUIRE{Functions $f, g: G \to \bbD(d)$}
    \FOR{$s := \bigl(d^{3/2}/\epsilon\bigr)^{\Theta(d^2)}$ times}
      \STATE Sample $U\in U(d)$ with respect to the (normalized) Haar measure of $U(d)$.
      \STATE Let $e$ be the estimation of $\dist(f, Ug U^*)$ obtained by applying Lemma~\ref{lem:estimate} with the error parameter $\epsilon^{2}/100$ and the confidence parameter $1/6s$.
      \STATE \textbf{if} $e < \epsilon^{2}/10$ \textbf{then} accept.
    \ENDFOR
    \STATE Reject.
  \end{algorithmic}
\end{algorithm}



Let $U\in U(d)$ be a random matrix sampled with respect to the Haar measure.
We diagonalize $U$ as $U=W \Lambda W^*$ where $W \in U(d)$ and $\Lambda = \diag (\lambda_{1}, \ldots \lambda_{d})$.
By the unitarity of $U$, the absolute value of each eigenvalue of $U$ is 1. Therefore we can write $\lambda_{i} = \exp (\iunit \theta_{i})$ for some $\theta_{i} \in [-\pi, \pi)$, which we call the \emph{phase} of $\lambda_{i}$.
We use the following proposition, which is Weyl's integral formula applied to $U(d)$. 
\begin{proposition}[\cite{brocker1985representations}]
  The distribution $\mu$ of the phases $\theta = (\theta_{1}, \ldots, \theta_{d})$ is
    $\mathrm{d}\mu (\theta ) = \frac{1}{Z_{d}}\prod_{i > j} |\lambda_{i} - \lambda_{j}|^{2} \mathrm{d}\theta$,
  where $Z_{d} := (2\pi)^{d}d!$ is a normalization constant and $\mathrm{d}\theta$ is a standard Euclid measure.
\end{proposition}

For $\epsilon > 0$, we write  $B_{U(d)} (\epsilon) = \{U \in U(d) \mid \|U - I_d\|_{F} \leq \epsilon \}$.
We need the following auxiliary lemma, which says that the set of a sufficiently large number of randomly chosen unitary matrices forms an ``$\epsilon$-net'' of unitary matrices with respect to the Frobenius norm.
\begin{lemma}\label{lem:unitary-epsilon-net}
  Let $U_{0}\in U(d)$ and $U$ be a random matrix sampled with respect to the Haar measure of $U(d)$.
  For sufficiently small $\epsilon > 0$, the probability $\Pr[\|U - U_0\|_F \leq \epsilon] \geq \delta_{\ref{lem:unitary-epsilon-net}}(\epsilon,d)$, where $ \delta_{\ref{lem:unitary-epsilon-net}}(\epsilon,d) = \bigl(\frac{\epsilon}{d^{3/2}}\bigr)^{O(d^{2})}$.
\end{lemma}
\begin{proof}
  Since the Haar measure is invariant under left multiplication, we can assume $U_{0} = I_{d}$ without loss of generality.
  Hence, we want to bound $\Pr[U \in B_{U(d)}(\epsilon)]$.
  Let $\lambda_1,\ldots,\lambda_d$ be eigenvalues of $U$ and $\theta_1,\ldots,\theta_d$ be corresponding phases.
  By the conjugate invariance of Frobenius norm, $U\in B_{U(d)}(\epsilon)$ iff $\sum_{i=1}^{d} |\lambda_{i}-1|^2 \leq \epsilon^{2}$.

  Suppose $\|\theta\|_{2} \leq 3\epsilon/4$. Since $\theta_{i}$'s are sufficiently small, we can expand as $\lambda_{i} = 1 + \iunit \theta_{i} + O(\theta_{i}^{2})$.
  Then $\sum_{i=1}^{d} |\lambda_{i}-1|^2 = \|\theta\|_{2}^{2} + O(\sum_{i=1}^d\theta_{i}^{3}) = \|\theta\|_{2}^{2} + \|\theta\|_{2}^{2}\max_{i}|\theta_i|  < \epsilon^{2}$.
  It implies the probability is bounded below by $\Pr[ U\in B_{U(d)}(\epsilon)] \geq \int_{B_{d}(3\epsilon/4)} \mathrm{d}\mu$ where $B_{d}(r)$ is the ball in $\bbR^d$ of radius $r$ centered at the origin.

  Let $\tilde{\epsilon} = \frac{\epsilon}{d^{3/2}}$, $\theta_{0} = \left[ 0, \tilde{\epsilon}, \ldots, (d-1)\tilde{\epsilon} \right]$, and $\theta_{1} = \left[\tilde{\epsilon}/3, 4 \tilde{\epsilon}/3, \ldots, (d-2/3)\tilde{\epsilon} \right]$.
  Note that $\|\theta_{0}\|_{2}^{2} \leq \epsilon^{2}/3$ and $\|\theta_{1}\|_{2}^{2} = \epsilon^{2} (6d^{2}-d+2)/18d^{3} < \epsilon^{2}/2$.
  Therefore, $B_{d}(3\epsilon/4)$ contains the $d$-dimensional hypercube $V = \{\theta \in \bbR^d \mid \theta_0 \leq \theta \leq \theta_1\}$, where we write $x \leq y$ if $x_i \leq y_i$ for each $i$.
  Note that if $\theta \in V$, $|\theta_i - \theta_j| \geq \tilde{\epsilon}/3$ for any $i \neq j$.
  Therefore, we have
  \begin{align*}
    \int_{B_{d}(3\epsilon/4)} \mathrm{d}\mu
    &\geq \int_{V} \mathrm{d}\mu \geq \frac{1}{Z_{d}} \int_{V} \prod_{i > j} |\lambda_{i} - \lambda_{j}|^{2} \mathrm{d}\theta
    \sim \frac{1}{Z_{d}}\int_{V} \prod_{i > j} |\theta_{i} - \theta_{j}|^{2} \mathrm{d}\theta  \tag{By the fact that $\epsilon$ is small enough} \\
    &\geq \frac{1}{Z_{d}} \int_{V} \prod_{i > j} \Bigl(\frac{\tilde{\epsilon}}{3}\Bigr)^2 \mathrm{d}\theta = \frac{1}{Z_{d}} \Bigl(\frac{\tilde{\epsilon}}{3}\Bigr)^{d(d-1)} \int_{V} \mathrm{d}\theta = \frac{1}{Z_{d}} \Bigl(\frac{\tilde{\epsilon}}{3}\Bigr)^{d^{2}}.
  \end{align*}
\end{proof}

Let $\Skew(d)$ be the set of $d$-dimensional skew-Hermitian matrices, i.e., $\Skew(d) = \{X \in M_{d}(\bbC) \mid X^* = -X  \}$.
Although the following lemma is an almost immediate consequence of the fact that $\Skew(d)$ is the Lie algebra of $U(d)$,
we prove it for completeness.
\begin{lemma}\label{lem:exponential-map}
  If $\epsilon > 0$ is small enough, then the image of the exponential map $\exp:\Skew (d) \to U(d)$ contains $B_{U(d)}(\epsilon)$. Furthermore, if $U \in B_{U(d)}(\epsilon)$, then we can choose $X\in \Skew(d)$ such that $\exp(X) = U$ and $\|X\|_{F} \leq 2\epsilon$.
\end{lemma}
\begin{proof}
  We diagonalize $U$ as $W\Lambda W^*$, where $W\in U(d)$ and $\Lambda = \diag (\lambda_{1}, \ldots, \lambda_{d})$.
  Let $\lambda_{i} = \exp (\iunit \theta_{i})$ for $\theta_{i} \in \bbR$ and $\mu_{i} = 1-\lambda_{i}$.
  For $z\in \bbC$, we define $\log (1+z) = \sum_{i=1}^{\infty} (-1)^{n+1} z^{n}/n$. This Taylor expansion converges for $|z| \leq 1$.
  Note that $\log \lambda_{i} = \iunit (\theta_{i} + 2\pi m_{i})$ for some $m_{i} \in \bbZ$.

  Since $\|U-I_{d}\|_{F} \leq \epsilon$, we can define $X = \log U$.
  Note that $U = \exp(X)$ as $z = \exp (\log z)$.
  We have the following formulas:
  \begin{align*}
    X & = \log U = \log W \Lambda W^* = W (\log \Lambda) W^*, \\
    \log \Lambda & = \diag (\log \lambda_{1}, \ldots, \log \lambda_{d}), \\
    \overline{\log \lambda_{i}} & = \overline{\iunit (\theta_{i} + 2\pi m_{i})} = -\iunit (\theta_{i} + 2\pi m_{i}) = -\log \lambda_{i}, \\
    (\log \Lambda)^* & = \diag (\overline{\log \lambda_{1}}, \ldots, \overline{\log \lambda_{d}}) = -\diag (\log \lambda_{1}, \ldots, \log \lambda_{d}) = - \log \Lambda.
  \end{align*}
  Combining these formulas, $X^* = W (\log \Lambda)^* W^* = -W (\log \Lambda) W^* = -X$, that is, $X \in \Skew (d)$.

  Next, we estimate the Frobenius norm of $X$.
  The condition $\|U-I_{d}\|_{F} \leq \epsilon$ implies $\sum_{i=1}^{d} |\mu_{i}|^{2} \leq \epsilon^{2}$.
  It implies $\mu_{i}$ is small enough. So we can expand $\log \lambda_{i} = \log (1 + \mu_{i}) = \mu_{i} + O(\mu_{i}^{2})$.
  Therefore $|\log \lambda_{i}|^{2} = |\mu_{i}|^{2} + O(\mu_{i}^{3})$ and we can conclude $\|X\|_{F}^{2} = \|W(\log \Lambda) W^*\|_{F}^{2} = \|\log \Lambda\|_{F}^{2} = \sum_{i=1}^{d} |\log \lambda_{i}|^{2} = \sum_{i=1}^{d} |\mu_{i}|^{2} + O(\mu_{i}^{3}) \leq 2\epsilon^{2}$.
\end{proof}

\begin{lemma}\label{lem:U-V->UAU-VAV}
  Let $A \in \bbD(d)$ be a matrix and $U, V\in U(d)$ be unitary matrices.
  If $\| U-V\|_F \leq \epsilon$ for sufficiently small $\epsilon > 0$, then $\| UAU^* - VAV^*\|_{F} \leq 3\epsilon$.
\end{lemma}
\begin{proof}
  Let $U_0 = V^*U$.
  By a direct calculation, $\|U-V\|_{F} = \|U_0 - I_d\|_{F}$ and $\| UAU^* - VAV^*\|_{F} = \| U_0AU_0^* - A\|_{F}$.
  Therefore, we can assume that $V = I_d$.
  By Lemma~\ref{lem:exponential-map}, there exists $X\in \Skew (d)$ such that $\|X\|_{F} \leq 1$ and $U = \exp (2\epsilon X)$. Note that $U^{\ast} = \exp (-2\epsilon X)$. Then,
  \begin{align*}
    U^{\ast} A U - A & = \exp(-2\epsilon X)A \exp(2\epsilon X) - A
     = (I_{d} - 2\epsilon X + O(\epsilon^{2})J) A (I_{d} + 2\epsilon X + O(\epsilon^{2})J) \\
    &  = 2\epsilon (AX - XA) + O(\epsilon^{2})J,
  \end{align*}
  where $O(\epsilon^2)J$ denotes a matrix with each entry having absolute value $O(\epsilon^2)$.
  We evaluate the Frobenius norm of the commutator as
  \[
    \|AX - XA\|_{F} \leq \|AX\|_{F} + \|XA\|_{F} = 2\|X\|_{F} \leq 2,
  \]
  where we use the triangle inequality of the Frobenius norm, and the assumption $\|A\|_F \leq 1$.
  Therefore,
  \[
    \|U A U^* - A\|_{F}^{2} \leq 4\epsilon^{2} \|AX - XA\|_{F}^{2} + O(\epsilon^{3}) \leq 8\epsilon^{2} + O(\epsilon^3) < 9\epsilon^{2}.
  \]
\end{proof}


Now, we prove Theorem~\ref{the:unitary}.
\begin{proof}[Proof of Theorem~\ref{the:unitary}]
  We assume that the all the estimations have succeeded, which happens with probability at least $5/6$ by the union bound.

  Suppose that $f$ and $g$ are unitary equivalent and $U_0 \in U(d)$ be the unitary matrix with $f(x) = U_0g(x)U_0^*$ for any $x \in G$.
  From Lemma~\ref{lem:unitary-epsilon-net}, we sample $U$ such that $\|U_{0} - U\|_{F} \leq \epsilon/10$ with probability at least $5/6$.
  From Lemma~\ref{lem:U-V->UAU-VAV}, we have $\E_{x\in G}[\|f(x) -U g(x) U^* \|_{F}^{2}] \leq 9\epsilon^{2}/100$.
  For such $U$, we obtain the estimation $e$ satisfies $e < \epsilon^{2}/10$.
  By the union bound, we accept with probability at least $2/3$.

  Suppose that $f$ and $g$ are $\epsilon$-far from being unitary equivalent.
  For every $U$ we sample,
  we have $\E_{x\in G}[\|f(x) -U g(x) U^* \|_{F}^{2}] \geq 4\epsilon^2$.
  Hence, the estimation $e$ satisfies $e > \epsilon^2/10$ and we accept.
  To summarize, we reject with probability at least $5/6>2/3$.
\end{proof}

\section*{Acknowledgements}
We thank Mitsuru Kusumoto and anonymous referees for comments that greatly improved the manuscript.

\bibliographystyle{abbrv}
\bibliography{main}

\begin{thebibliography}{10}

\bibitem{Allender:1998hf}
E.~Allender, J.~Jiao, M.~Mahajan, and V.~Vinay.
\newblock Non-commutative arithmetic circuits: depth reduction and size lower
  bounds.
\newblock {\em Theoretical Computer Science}, 209(1-2):47--86, 1998.

\bibitem{Alon:2013ci}
N.~Alon and S.~Lovett.
\newblock Almost $k$-wise vs. $k$-wise independent permutations, and uniformity
  for general group actions.
\newblock {\em Theory of Computing}, 9(1):559--577, 2013.

\bibitem{Arora:1998dga}
S.~Arora, C.~Lund, R.~Motwani, M.~Sudan, and M.~Szegedy.
\newblock Proof verification and the hardness of approximation problems.
\newblock {\em Journal of the ACM}, 45(3):501--555, 1998.

\bibitem{Bellare:1996iv}
M.~Bellare, D.~Coppersmith, J.~Hastad, M.~Kiwi, and M.~Sudan.
\newblock Linearity testing in characteristic two.
\newblock {\em IEEE Transactions on Information Theory}, 42(6):1781--1795,
  1996.

\bibitem{BenOr:2007cg}
M.~Ben-Or, D.~Coppersmith, M.~Luby, and R.~Rubinfeld.
\newblock Non-abelian homomorphism testing, and distributions close to their
  self-convolutions.
\newblock {\em Random Structures {\&} Algorithms}, 32(1):49--70, 2007.

\bibitem{BenSasson:2003dn}
E.~Ben-Sasson, M.~Sudan, S.~Vadhan, and A.~Wigderson.
\newblock Randomness-efficient low degree tests and short {PCPs} via
  epsilon-biased sets.
\newblock In {\em Proceedings of the 35th annual ACM symposium on Theory of
  computing (STOC)}, pages 612--621, 2003.

\bibitem{Berman:2014kg}
P.~Berman, S.~Raskhodnikova, and G.~Yaroslavtsev.
\newblock {$L_p$-testing}.
\newblock In {\em Proceedings of the 46th Annual ACM Symposium on Theory of
  Computing (STOC)}, pages 164--173, 2014.

\bibitem{Bhattacharyya:2013gu}
A.~Bhattacharyya.
\newblock Guest column: on testing affine-invariant properties over finite
  fields.
\newblock {\em ACM SIGACT News}, 44(4):53--72, 2013.

\bibitem{Bhattacharyya:2013ii}
A.~Bhattacharyya, E.~Fischer, H.~Hatami, P.~Hatami, and S.~Lovett.
\newblock Every locally characterized affine-invariant property is testable.
\newblock In {\em Proceedings of the 45th Annual ACM Symposium on Theory of
  Computing (STOC)}, pages 429--436, 2013.

\bibitem{Bhattacharyya:2012ud}
A.~Bhattacharyya, E.~Fischer, and S.~Lovett.
\newblock Testing low complexity affine-invariant properties.
\newblock {\em Proceedings of the 24th Annual ACM-SIAM Symposium on Discrete
  Algorithms (SODA)}, pages 1337--1355, 2012.

\bibitem{Bhattacharyya:2010kw}
A.~Bhattacharyya, S.~Kopparty, G.~Schoenebeck, M.~Sudan, and D.~Zuckerman.
\newblock Optimal testing of {Reed}-{Muller} codes.
\newblock In {\em Proceedings of the 51st Annual IEEE Symposium on Foundations
  of Computer Science (FOCS)}, pages 488--497, 2010.

\bibitem{Biane:2003hw}
P.~Biane.
\newblock Characters of symmetric groups and free cumulants.
\newblock {\em Asymptotic Combinatorics with Applications to Mathematical
  Physics}, 1815:185--200, 2003.

\bibitem{Yoshida:2012dm}
E.~Blais, A.~Weinstein, and Y.~Yoshida.
\newblock Partially symmetric functions are efficiently isomorphism-testable.
\newblock In {\em Proceedings of the 53rd Annual IEEE Symposium on Foundations
  of Computer Science (FOCS)}, pages 551--560, 2012.

\bibitem{Blum:1993cn}
M.~Blum, M.~Luby, and R.~Rubinfeld.
\newblock Self-testing/correcting with applications to numerical problems.
\newblock {\em Journal of Computer and System Sciences}, 47(3):549--595, 1993.

\bibitem{brocker1985representations}
T.~Br{\"o}cker and T.~Dieck.
\newblock {\em Representations of Compact Lie Groups}.
\newblock Graduate Texts in Mathematics. Springer, 1985.

\bibitem{Green:2005iv}
B.~Green.
\newblock A {Szemer{\'e}di}-type regularity lemma in abelian groups, with
  applications.
\newblock {\em Geometric and Functional Analysis}, 15(2):340--376, 2005.

\bibitem{Hatami:2013ux}
H.~Hatami and S.~Lovett.
\newblock Estimating the distance from testable affine-invariant properties.
\newblock In {\em Proceedings of the 54th Annual IEEE Symposium on Foundations
  of Computer Science (FOCS)}, pages 237--242, 2013.

\bibitem{Huang:2009vk}
J.~Huang, C.~Guestrin, and L.~Guibas.
\newblock Fourier theoretic probabilistic inference over permutations.
\newblock {\em The Journal of Machine Learning Research}, 10:997--1070, 2009.

\bibitem{Kassabov:2007ki}
M.~Kassabov.
\newblock Symmetric groups and expander graphs.
\newblock {\em Inventiones Mathematicae}, 170(2):327--354, 2007.

\bibitem{Kaufman:2008jm}
T.~Kaufman and M.~Sudan.
\newblock Algebraic property testing: the role of invariance.
\newblock In {\em Proceedings of the 40th Annual ACM Symposium on Theory of
  Computing (STOC)}, pages 403--412, 2008.

\bibitem{Kerov:2000aa}
S.~V. Kerov.
\newblock Talk at {IHP} conference.
\newblock 2000.

\bibitem{Alexander2008LieGroup}
A.~Kirillov.
\newblock {\em Introduction to Lie groups and Lie algebras}.
\newblock Cambridge University Press, 2008.

\bibitem{ODonnell:2015va}
R.~O'Donnell and J.~Wright.
\newblock Quantum spectrum testing.
\newblock In {\em Proceedings of the 47th ACM Symposium on Theory of Computing
  (STOC)}, 2015.

\bibitem{Raz:1995ib}
R.~Raz and B.~Spieker.
\newblock On the ``log rank''-conjecture in communication complexity.
\newblock {\em Combinatorica}, 15(4):567--588, 1995.

\bibitem{Ron:2009hh}
D.~Ron.
\newblock Algorithmic and analysis techniques in property testing.
\newblock {\em Foundations and Trends{\textregistered} in Theoretical Computer
  Science}, 5(2):73--205, 2009.

\bibitem{Rubinfeld:2006jy}
R.~Rubinfeld.
\newblock On the robustness of functional equations.
\newblock {\em SIAM Journal on Computing}, 28(6):1972--1997, 2006.

\bibitem{Rubinfeld:1996um}
R.~Rubinfeld and M.~Sudan.
\newblock Robust characterizations of polynomials with applications to program
  testing.
\newblock {\em SIAM Journal on Computing}, 25(2):252--271, 1996.

\bibitem{Shpilka:2006fp}
A.~Shpilka and A.~Wigderson.
\newblock Derandomizing homomorphism testing in general groups.
\newblock {\em SIAM Journal on Computing}, 36(4):1215--1230, 2006.

\bibitem{Stanley:2003tv}
R.~P. Stanley.
\newblock Irreducible symmetric group characters of rectangular shape.
\newblock {\em S{\'e}minaire Lotharingien de Combinatoire}, 50, 2003.

\bibitem{steinberg2011representation}
B.~Steinberg.
\newblock {\em Representation Theory of Finite Groups: An Introductory
  Approach}.
\newblock Springer, 2011.

\bibitem{Weyl:1927dn}
H.~Weyl.
\newblock Quantenmechanik und gruppentheorie.
\newblock {\em Zeitschrift f{\"u}r Physik}, 46(1-2):1--46, 1927.

\bibitem{Wimmer:2010jw}
K.~Wimmer.
\newblock Agnostically learning under permutation invariant distributions.
\newblock In {\em Proceedings of the 51st Annual IEEE Symposium on Foundations
  of Computer Science (FOCS)}, pages 113--122, 2010.

\bibitem{Yoshida:2014tq}
Y.~Yoshida.
\newblock A characterization of locally testable affine-invariant properties
  via decomposition theorems.
\newblock In {\em Proceedings of the 46th Annual ACM Symposium on Theory of
  Computing (STOC)}, pages 154--163, 2014.

\end{thebibliography}

\end{document}